\documentclass[10pt,a4paper]{article}
\usepackage{fullpage,nopageno}
\usepackage{graphicx,amssymb,amsmath,xspace,amsfonts,tabularx}
\usepackage{clrscode}

\advance\textheight 3\baselineskip

\newtheorem{theorem}{Theorem}

\newenvironment{proof}{Proof:}{\qed}
\newenvironment{denseitems}{\list{$\bullet$}{\itemsep0pt\parsep0pt}}{\endlist}

\def\squareforqed{\hbox{\rlap{$\sqcap$}$\sqcup$}}
\def\qed{\ifmmode\squareforqed\else{\unskip\nobreak\hfil
\penalty50\hskip1em\null\nobreak\hfil\squareforqed
\parfillskip=0pt\finalhyphendemerits=0\endgraf}\fi}

\def\tin{\textsc{tin}\xspace}
\def\tins{\textsc{tin}s\xspace}

\def\gis{\textsc{gis}\xspace}
\def\lidar{\textsc{lidar}\xspace}
\def\io{\textsc{i/o}\xspace}
\def\ios{\textsc{i/o}'s\xspace}
\def\tsm{\textsc{TerraStream}\xspace}
\def\tfl{\textsc{TerraFlow}\xspace}

\def\sort{\mathit{Sort}}
\def\Sort(#1){({#1}/B) \log_{M/B} ({#1}/B)}
\def\scan{\mathit{Scan}}

\title{Simple I/O-efficient flow accumulation on grid terrains}

\author{%
Herman~Haverkort\thanks{TU Eindhoven, the Netherlands, cs.herman@haverkort.net}
\hspace{.3cm}
Jeffrey Janssen\thanks{Realworld Systems, Culemborg, the Netherlands, twilight@nmotion.nl}
}

\date{}

\begin{document}
\maketitle

\begin{abstract}
The \emph{flow accumulation} problem for grid terrains takes as input a matrix of flow directions, that specifies for each cell of the grid to which of its eight neighbours any incoming water would flow. The problem is to compute, for each cell $c$, from how many cells of the terrain water would reach $c$. We show that this problem can be solved in $O(\scan(N))$ \ios for a terrain of $N$ cells. Taking constant factors in the \io-efficiency into account, our algorithm may be an order of magnitude faster than the previously known algorithm that is based on time-forward processing and needs $O(\sort(N))$ \ios.
\end{abstract}

\section{Introduction}
Current remote-sensing technology provides high-resolution terrain data for geographic information systems (\gis). In particular, with \lidar technology one can get detailed elevation models of the earth surface. These are usually made available in the form of a triangulated irregular network (\tin) or a grid: a matrix with elevation values for points in a regular grid on the surface of the earth. One important application of such models are hydrological studies: analysing the flow of water on a terrain, for example to study the effects of possible human intervention or to estimate risks of erosion.

The grid models may be as large as several dozen gigabytes so that they do not fit in the main memory of a computer at once. Therefore hydrological analysis requires efficient algorithms that scale well and are designed to minimise the swapping of data between main memory and disk. Throughout the current decade Arge et al.\ have designed and published such algorithms for the following pipeline of computations on terrain data:\begin{denseitems}
\item constructing an elevation model from raw elevation samples~\cite{cloud,terrastream};
\item (partial) \emph{flooding:} eliminating spurious depressions from the model, so that streams do not appear to be interrupted by virtual dams that result from sampling errors~\cite{unionfind,terraflow,terrastream};
\item \emph{flow routing:} determining in what direction water could flow on each point of the terrain, in such a way that from every point of the terrain water would follow a non-ascending path to the lowest point of the watershed that contains that point~\cite{terraflow,terrastream};
\item \emph{flow accumulation:} computing for every point of the terrain the size of the region from which water flows to that point~\cite{terraflow,gridproblems,terrastream};
\item \emph{hierarchical watershed labelling:} giving each point of the terrain a label that indicates its position in the watershed hierarchy~\cite{pfafstetter,terrastream}.
\end{denseitems}
Arge et al.\ implemented these algorithms in the form of a system called \tsm~\cite{terrastream}.
While the published algorithms have been shown to be very effective, in this paper we will show that for some of these algorithms alternatives exist that may process grid models (but not \tins) up to an order of magnitude faster.

\paragraph*{Analysing I/O-efficiency}
In this paper we analyse the efficiency of algorithms with the standard model that was defined by Aggarwal and Vitter~\cite{iomodel}. In this model, a computer has a memory of size $M$ and a disk of unbounded size. The disk is divided into blocks of size $B$. Data is transferred between memory and disk by transferring complete blocks: transferring one block is called an ``\io''. Algorithms can only operate on data that is currently in main memory; to access the data in any block that is not in main memory, it first has to be copied from disk. If data in the block is modified, it has to be copied back to disk later, at the latest when it is evicted from memory to make room for another block. In this paper we will be concerned with the \io-efficiency of algorithms: the number of \ios they need as a function of the input size $N$, the memory size $M$, and the block size $B$. As a point of reference, scanning $N$ consecutive records from disk takes $\scan(N) = \Theta(N/B)$ \ios; sorting takes $\sort(N) = \Theta(\Sort(N))$ \ios in the worst case~\cite{iomodel}. It is sometimes assumed that $M = \Omega(B^2)$.

We distinguish \emph{cache-aware} algorithms and \emph{cache-oblivious} algorithms. Cache-aware algorithms may use knowledge of $M$ and $B$ (and to some extent even control $B$) and they may use this to control which blocks are kept in memory and which blocks are evicted. Cache-oblivious algorithms do not know $M$ and $B$ and cannot control which blocks are kept in memory: the caching policy is left to the hardware and the operating system. Nevertheless cache-oblivious algorithms can often be designed and proven to be \io-efficient~\cite{flpr-coa-99}. The idea is to design the algorithm's pattern of access to locations in input and output files such that effective caching is achieved by any reasonable general-purpose caching policy (such as least-recently-used replacement) for any values of $M$ and $B$. As a result, any bounds that can be proven on the \io-efficiency of a cache-oblivious algorithm do not only apply to the transfer of data between disk and main memory, but also to the transfer of data between main memory and the various levels of smaller caches. However, in practice cache-oblivious algorithms cannot always match the performance of cache-aware algorithms that are tuned to specific values of $M$ and $B$.

\paragraph*{Our results}
We present and analyse several algorithms for flow accumulation on grid terrains:\begin{denseitems}
\item an ``\io-na\"ive'' algorithm that processes the grid row by row;
\item a variant of the above that processes the row-by-row data in so-called Z-order;
\item a variant of the above that processes data that is stored in Z-order;
\item a cache-aware \io-efficient algorithm based on separators;
\item a variant of the above that processes data that is stored in Z-order;
\item a cache-oblivious variant of the above;
\item and for comparison: the algorithm based on time-forward processing from Arge et al.~\cite{terraflow,gridproblems,terrastream}
\end{denseitems}
Our results are summarised in Table~\ref{tab:results}.
We find that the last algorithm induces a significant overhead from sorting the input cells into topological order. This is not only because the sorting itself takes \io, but also because sorting the cells makes it necessary to store the coordinates in the grid for each cell: the coordinates are needed to sort the cells back into row-by-row order when the computation is complete. Our new algorithms avoid this overhead: they do not sort. The best of our algorithms are asymptotically more efficient than time-forward processing (under mild assumptions), and because the coordinates of a grid cell can always be deduced from its location in the file, there is no overhead from storing coordinates. In practice the difference in performance may be up to an order of magnitude, and preliminary experimental results indicate that our algorithms are fast indeed. For comparison: the authors of \tsm reported 455 minutes for flow accumulation by time-forward processing, working on a grid of similar size and hardware that appeared to be slightly faster than ours.
\begin{table}
\def\arraystretch{1.25}
\begin{tabularx}{\hsize}{|Xr|lr|r|r|}
\hline
algorithm & files & \multicolumn{2}{c|}{asymptotic number of \ios} & \multicolumn{1}{c|}{\io-volume /} & running time\\
          & & worst case & realistic & \multicolumn{1}{c|}{(in+output)} & \multicolumn{1}{c|}{$N = 3.5 \cdot 10^9$}\\
\hline\hline
\multicolumn{6}{|l|}{flow accumulation algorithms}\\
\hline
``na\"ive'' row-by-row scan & \textsf{R}                  & $O(N)$                           & $O(N/\sqrt B)${ $^{\textsf{C\hphantom{T}}}$} &         & 111 min.\hphantom{*}\\
``na\"ive'' Z-order-scan & \textsf{R}     &$O(N)$                           & $O(\scan(N))${ $^{\textsf{CT}}$} &         &     \\
``na\"ive'' Z-order-scan & \textsf{Z}        & $O(N/\sqrt B)$                   & $O(\scan(N))${ $^{\textsf{C\hphantom{T}}}$} &         &  41 min.\hphantom{*}\\
cache-aware separators & \textsf{R}                   & \multicolumn{2}{c|}{$O(\scan(N))${ $^{\textsf{AT}}$}} & 2.0--6.6{ $^{\textsf{A}}$} &  39 min.\hphantom{*}\\
cache-aware separators & \textsf{Z}   & \multicolumn{2}{c|}{$O(\scan(N))${ $^{\textsf{A\hphantom{T}}}$}} &     1.1{ $^{\textsf{A}}$}  & \\
cache-oblivious separators & \textsf{R}              & \multicolumn{2}{c|}{$O(\scan(N))${ $^{\textsf{T\hphantom{A}}}$}}  & & \\
cache-oblivious separators & \textsf{Z}              & \multicolumn{2}{c|}{$O(\scan(N))${ $^{\textsf{T\hphantom{A}}}$}}  & & 118 min.\hphantom{*}\\
\hline
time-forward proc.~\cite{terraflow,gridproblems,terrastream} & \textsf{RZ} & \multicolumn{2}{c|}{$O(\sort(N))${ $^{\textsf{\hphantom{TA}}}$}}                    &   7.8--32.0{ $^{\textsf{A}}$} & \\
\hline\hline
\multicolumn{6}{|l|}{conversion to/from Z-order}\\
\hline
\multicolumn{2}{|l|}{Z-order scan} & \multicolumn{2}{c|}{$O(\scan(N))${ $^{\textsf{T\hphantom{A}}}$}} &  & 88 min.\hphantom{*}\\
\multicolumn{2}{|l|}{row-by-row scan} & \multicolumn{2}{c|}{$O(\scan(N))${ $^{\textsf{\rlap{S}\hphantom{TA}}}$}} & 2.0\hphantom{ $^{\textsf{A}}$} & \\
\multicolumn{2}{|l|}{by sorting} & \multicolumn{2}{c|}{$O(\sort(N))${ $^{\textsf{\hphantom{TA}}}$}} & 4.0--6.0{ $^{\textsf{A}}$} & \\
\hline
\end{tabularx}
\caption{%
\textsf{R}: input/output files in row-by-row order;\quad
\textsf{Z}: input/output files in Z-order;\quad
\textsc{RZ}: results hold for both formats.\quad
\textsf{A}: requires knowledge of $M$ and/or knowledge or tuning of~$B$;\quad
\textsf{C}: requires the \emph{confluence assumption} (see Section~\ref{sec:confluence});\quad
\textsf{S}: requires that $\Theta(\sqrt B)$ rows of the input fit in memory;\quad
\textsf{T}: requires tall-cache assumption ($M \geq c B^2$ for some constant~$c$).\hfill\break
The \io-volume is the number of bytes transferred relative to the size of the input and output files, for a range of reasonable values for $N$, $M$, and $B$ (see Section~\ref{sec:parameters}).}
\label{tab:results}
\end{table}
This suggests that, while the time-forward processing algorithms are more flexible (they are easier to adapt to triangulated terrains and to so-called multiple-flow-direction models), their flexibility comes at a price when processing grid terrains. Simple alternatives such as our algorithms based on Z-order or our cache-aware algorithm are therefore worth to be considered for practical applications.

In the following section we discuss some preliminaries: we introduce the \emph{confluence assumption} which is helpful when analysing the \io-efficiency of algorithms on realistic inputs, and we discuss how to traverse a grid in Z-order and how to convert a grid file from row-by-row order to Z-order and back. In the next section we present and analyse our flow accumulation algorithms and experimental results.
In Section~\ref{sec:otherstages} we discuss possible applications of our ideas to other stages of the hydrological analysis pipeline.

\section{Preliminaries}

\subsection{The parameters of the I/O-model}
\label{sec:parameters}
Sometimes we will require the \emph{tall-cache assumption} ($M \geq c B^2$ for some constant $c$); when this is the case we will indicate this explicitly.
When estimating \io-volumes, we will explicitly consider the scenarios $M = 2^{30}$ bytes, $B = 2^{16}$ bytes, and $M = 2^{30}$ bytes, $B = 2^{14}$ bytes, in both cases with an input grid of size $N = 2^{32}$ cells (this is roughly the size of the largest grid we experimented with).

\subsection{The confluence assumption}\label{sec:confluence}
In the next section we will see that some algorithms would perform very poorly when given a grid of $\sqrt N \times \sqrt N$ cells in which a river flows from the right to the left, making $c_1 \sqrt N$ meanders with an amplitude of $c_2 \sqrt N$ cells (see Figure~\ref{fig:meanders}, left), for some constants $c_1$ and $c_2$. However, it does not seem to make all that much sense to take such cases into account in asymptotic analysis: that would suggest that with growing resolution and size of our grids, rivers would get more and larger meanders. But the major features of real drainage networks would never change just because our grids get bigger. To capture this intuition, we pose the \emph{confluence assumption}, defined as follows.

For any square $Q$ of $d \times d$ cells, let $Q'$ be the square of $3d \times 3d$ cells centered on $Q$.
Let the \emph{first far cells} of $Q$ be the cells $c'$ on the boundary of $Q'$ such that from at least one cell $c$ of $Q$ water reaches $c'$ before leaving $Q'$. The \emph{confluence assumption} states that there is a constant $\gamma$, independent of $N$ and $d$, such that any square $Q$ has at most $\gamma$ different first far cells (see Figure~\ref{fig:confluence}, middle). We call $\gamma$ the \emph{confluence constant}.
Informally, the confluence assumption says that the water that flows from any square region will collect into a small number of river valleys before it gets very far.

We will use the confluence assumption in the analysis of some of our algorithms.
\begin{figure}
\centering\includegraphics[width=\hsize]{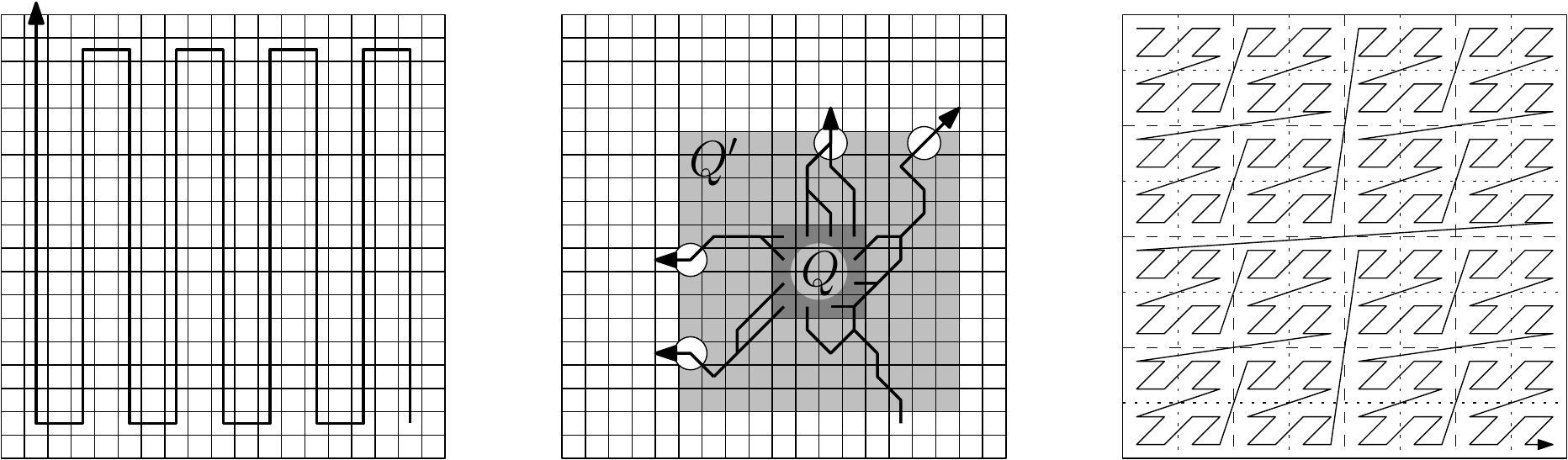}
\caption{Left: an unrealistic grid with a river with $\Theta(\sqrt N)$ meanders with an amplitude of $\Theta(\sqrt N)$.
Middle: a realistic grid. The white encircled cells are the first far cells of $Q$; it has only few first far cells.
Right: Z-order.}
\label{fig:meanders}\label{fig:confluence}\label{fig:zorder}
\end{figure}

\subsection{Storing grids in Z-order}\label{sec:Zorder}
Some of our algorithms will benefit from processing or storing the grid cells in Z-order. This order is defined as follows.
Say our input is a grid $G$ of height $h$ and width $w$. This grid is contained in a matrix $G'$ of $k \times k$ cells, where $k$ is the smallest power of 2 such that $k \geq h$ and $k \geq w$. Let the first cell of $G$ (in the upper left corner) also be the first cell of $G'$. The matrix $G'$ has four quadrants of $k/2 \times k/2$ cells. When storing $G$ in Z-order, we first store the part of $G$ in the upper left quadrant of $G'$, then the part of $G$ in the upper right quadrant of $G'$, then the part of $G$ in the lower left quadrant of $G'$, and finally the part of $G$ in the lower right quadrant of $G'$. Within each quadrant, the cells of $G$ are ordered recursively in the same way (see Figure~\ref{fig:zorder}, right).

\paragraph{Computing Z-order coordinates from row/column coordinates and vice versa}

Assume, for now, that $k = h = w$, and consider the cell $c$ in row $y$ and column $x$ of the grid. Let $y_1y_2...y_m$ and $x_1x_2...x_m$ be the binary representations of $y$ and $x$ (where $m = \log_2 k$). Then the position of $c$ in a row-by-row file is $y_1y_2...y_mx_1x_2...x_m$, and its position in a Z-order file is $y_1x_1y_2x_2...y_mx_m$. It is therefore quite easy to obtain, for any cell, its position in a row-by-row file from its position in a Z-order file and vice versa.

When the input grid is not square or when its height and width are not integral powers of two, the conversion is slightly more difficult, but in practice it can still be done efficiently as follows. Let $G$ be a grid of size $h \times w$ stored in Z-order. Then $G'$ has size $2^m \times 2^m$, where $m = \max(\lceil\log h\rceil,\lceil\log w\rceil)$. The sequence of cells of the matrix $G'$ in Z-order can now be divided into $\Theta(w+h)$ subsequences $D_1,N_1,D_2,N_2,...,D_t,N_t$, such that each sequence $D_i$ consists of cells in $G$ and each sequence $N_i$ consists of cells outside~$G$. We can now construct two arrays $F[1..t]$ and $Z[1..t]$, where $F[i] = \sum_{j=1}^{i-1} |D_j|$ and $Z[i] = \sum_{j=1}^{i-1} (|D_j| + |N_j|)$. Thus $F[i]$ and $Z[i]$ store, for each segment $D_i$, its offset in the file (which only stores the cells in $G$) and its offset in the Z-order traversal of $G'$, respectively.

An offset $p$ in the Z-order file can now be converted to row and column coordinates as follows: find the highest $i$ such that $F[i] \leq p$, compute $y_1x_1y_2x_2...y_mx_m = Z[i] + (p - F[i])$, and extract the row and column coordinates $y_1y_2...y_m$ and $x_1x_2...x_m$. Row and column coordinates can be converted to an offset in the Z-order file in a symmetric way.

Note that $F$ and $Z$ can be computed without reading the input grid. For all reasonable grid sizes $F$~and $Z$ easily fit in main memory. When this is not the case, the number of \ios incurred by swapping blocks of $F$ and $Z$ is always dominated by the number of \ios incurred by accessing the cells whose offsets or coordinates are being computed. The costs of converting coordinates can therefore be ignored in the \io-efficiency analysis of the algorithms in this paper.

From a practical point of view one should also consider the effort involved in the bit manipulations that are needed to extract $y_1y_2...y_m$ and $x_1x_2...x_m$ from $y_1x_1y_2x_2...y_mx_m$, and to find the Z-order offsets of the neighbours of a given cell $y_1x_1y_2x_2...y_mx_m$. This can also be done efficiently with a method called \emph{dilated arithmetic} and another set of look-up tables of size $\Theta(w + h)$~\cite{bitmanipulation}.

\paragraph{Converting grids from/to Z-order efficiently}
To convert a grid from row-by-row order to Z-order, we consider two very simple cache-oblivious algorithms: (A) read each cell from the row-by-row file and write it to the Z-order file, going through the grid in Z-order, and (B) read each cell from the row-by-row file and write it to the Z-order file, going through the grid row by row.

Algorithm (A), the Z-order scan, is efficient when we assume a tall cache. Then there is a $z = \Omega(B)$ such that $z$ is a power of two, a square of $z^2$ cells occupies $O(z^2/B)$ blocks of the row-by-row file on disk, and these blocks fit in memory.
The algorithm processes roughly $N / z^2$ such squares that each form a contiguous section of the Z-order file. For each such section the necessary blocks from the row-by-row file can be read and kept in memory until the section is completely processed. The total number of \ios for the conversion thus becomes at most $O(N/z^2 \cdot z^2/B) = O(\scan(N))$ for reading plus $O(\scan(N))$ for writing, making a total of $O(\scan(N))$.

If $\sqrt{B/s}$ rows of the grid fit in memory, where $s$ is the number of bytes per cell, then algorithm (B) is even more efficient. Algorithm (B), the row-by-row scan, would read every input block (from the row-by-row file) once and it would be able to keep every output block (which is $\sqrt{B/s}$ rows high) in memory until all of its cells have been read. Thus every output-block only needs to be written to disk once. In practice the assumption that $\sqrt{B/s}$ rows of the grid fit in memory seems reasonable; therefore we assume that the \io-volume of the conversion from row-by-row to Z-order is twice the input size. When $\sqrt{B/s}$ rows do not fit in memory, algorithm (B) may cause every output block to be reloaded $\sqrt{B/s}$ times, so that the algorithm uses $\Theta(N/\sqrt B)$ \ios.

Alternatively, one may adapt merge sort to convert a row-by-row file to Z-order in $\Theta(\sort(N))$ \ios, making two or three read/write passes over the input in practice.

To convert a grid from Z-order to row-by-row order, the same algorithms can be used while substituting reading for writing and vice versa.

\section{Algorithms for flow accumulation}

We assume that the input to the flow accumulation problem consists of a file $\id{FlowDir}$ that contains one byte for each cell, stored row-by-row. The byte for cell $c$ indicates the \emph{flow direction} of $c$: to which of the eight neighbours of $c$ (if any) any water that arrives at $c$ will flow. This neighbour is called the \emph{out-neighbour} of $c$. The required output is a file $\id{FlowAcc}$ with eight bytes per cell, stored row-by-row, that specify each cell's \emph{flow accumulation}. The flow accumulation of a cell $c$ is the number of cells that lie upstream of $c$, including $c$ itself. We can picture the flow accumulation of $c$ as the total number of units of rain that arrive at or pass through $c$ when each cell receives one unit of rain from the sky, which then flows down over the surface from cell to cell, following the flow directions, until it arrives at a cell that does not have an out-neighbour.

In the following we describe several algorithms for the flow accumulation problem and also discuss their \io-volume (number of \ios times block size) divided by the total size of the input and the output files. With input and output as defined above, the input+output size is 9 bytes times the number of cells in the grid.

\subsection{An I/O-na\"ive flow accumulation algorithm}\label{sec:naiveacc}

Our first algorithm is just clever enough to run in linear time, but utterly na\"ive from an \io point of view.
In this algorithm, let $c$ be the index (offset in the input file) of a cell in $\id{FlowAcc}$, and let $\id{out}(c)$ be the index of the neighbour of cell $c$ to which the incoming water of $c$ flows. Note that $\id{out}(c)$ can be computed from $c$ and the flow direction stored for $c$ in $\id{FlowDir}$.

\begin{codebox}
\Procname{$\proc{NaiveAccumulation}$}
\li Create a file $\id{FlowAcc}$ with flow value (initially 1) and marking bit (initially not set) for each cell
\li \For each cell $c$ (in row-by-row order)
\li \Do  \If $c$ is not marked
\li      \Then  $d \gets c$
\li             \While all in-neighbours of $d$ are marked and $d$ has an out-neighbour $\id{out}(d)$
\li             \Do    mark $d$
\li                    $\id{FlowAcc}[\id{out}(d)] \gets \id{FlowAcc}[\id{out}(d)] + \id{FlowAcc}[d]$
\li                    $d \gets \id{out}(d)$
                \End
         \End
    \End
\end{codebox}

The algorithm goes through all cells and gives each of them one unit of water; this water is then forwarded downstream until a cell is reached that is still waiting for incoming water from some of its neighbours. The accumulated flow then waits there until the cell has received the incoming flow from all neighbours, so that for each cell, all incoming flow is forwarded downstream in a single operation. Thus the algorithm runs in linear time.

However, the \io-behaviour of this simple algorithm can be quite bad: consider processing the terrain shown in Figure~\ref{fig:meanders} (left) in row-by-row order. None of the cells on the meandering river can have their flow forwarded downstream until the scan reaches the last cell in the lower right corner. At that point the while-loop of lines 5 to~8 will follow the whole river back to the upper left corner, possibly requiring one \io almost every step of the way. The worst-case \io-efficiency is therefore $\Theta(N)$. Fortunately, under the confluence assumption the situation does not look so grim:

\begin{theorem}\label{NaiveAccumulationRonR}
Under the confluence assumption algorithm \textsc{NaiveAccumulation} uses $O(N/\sqrt B)$ I/O's.
\end{theorem}
\begin{proof}
Consider the grid to consist of square subgrids of $\sqrt{B} \times \sqrt{B}$ cells.
While running the algorithm, we maintain that at line 3, the subgrid that contains $c$ and the eight subgrids around it are in memory. When in line 5--8 we access a cell $d$ that is currently not in memory, we load the subgrid that contains $d$ and the subgrids around it into memory; upon termination of this loop we reload the subgrids that were in memory before entering the loop.

Ignoring the loop in line 5--8, each disk block of the row-by-row files is loaded into memory at most $3\sqrt{B}$ times, causing $O(N/\sqrt B)$ \ios in total. In line 5--8, the nine subgrids currently in memory are replaced (and possibly reloaded afterwards) only when we reach and finish a cell on the boundary of the nine subgrids after following a path that started from the subgrid in the middle. By the confluence assumption this happens at most a constant number of times for each group of $3 \times 3$ subgrids. Since there are $O(N/B)$ such groups (one centered on each subgrid) and each is stored across at most $O(\sqrt B)$ blocks on disk, the algorithm takes $O(N/\sqrt B)$ \ios in total.
\end{proof}

\subsection{Flow accumulation in Z-order with row-by-row files}

We can make algorithm \textsc{NaiveAccumulation} slightly smarter by changing the loop in line 2 to go through all cells in Z-order. This does not change the worst-case \io-efficiency of the algorithm (extreme rivers could still cause $\Theta(N)$ \ios). However, we would expect significantly better performance in practice:

\begin{theorem}\label{NaiveAccumulationZonRrealistic}
Under the confluence assumption and with a tall cache of size $M \geq c B^2$, algorithm \textsc{NaiveAccumulation}, running the loop on line 2 in Z-order, needs only $O(\scan(N))$ I/O's.
\end{theorem}
\begin{proof}
By the tall-cache assumption, there is a $z = \Theta(B)$ such that $z$ is a power of two and the blocks containing a section of $3z \times 3z$ cells fit in memory. Now consider the grid to consist of square subgrids of $z \times z$ cells. Thus each subgrid contains $\Theta(B^2)$ cells and is stored in $\Theta(B)$ blocks in the row-by-row files. The idea is again to keep the subgrid of the current cell in memory, together with the eight subgrids around it.

Ignoring the loop in line 5--8, each group of $3 \times 3$ subgrids is loaded into memory once: each group is loaded when the Z-order scan of line 2 enters the subgrid $Q$ in the middle of the group, and the Z-order scan then traverses $Q$ completely before proceeding to the next subgrid. Thus each group is loaded at most once and each subgrid is loaded at most nine times, causing $\Theta(N/B^2 \cdot B) = \Theta(\scan(N))$ \ios. Following the same analysis as before, lines 5--8 cause each group of $3 \times 3$ subgrids to be evicted from memory at most a constant number of times, causing $O(\scan(N))$ \ios as well.
\end{proof}

\subsection{Flow accumulation in Z-order with Z-order files}

Algorithm NaiveAccumulation can be improved further by not only processing the cells in Z-order, but also using input in Z-order and producing output in Z-order. We now get the following:

\begin{theorem}
Algorithm \textsc{NaiveAccumulation}, running the loop on line 2 in Z-order and working on files in Z-order, needs only $O(N/\sqrt B)$ I/O's in the worst case. Under the confluence assumption, the required number of I/O's is $O(\scan(N))$.
\end{theorem}
\begin{proof}
As in the proof of Theorem~\ref{NaiveAccumulationRonR}, we consider the grid to be divided into subgrids of size $\sqrt B \times \sqrt B$. In Z-order files, these occupy only a constant number of blocks each. Such a subgrid $Q$ is loaded into memory when we access a cell on its boundary, or when it is reloaded after completing the loop in lines 5--8 of the algorithm. In the latter case we can charge the reloading of block $Q$ to the access to the boundary cell of $Q$ that was last accessed before $Q$ was evicted from memory. Since each subgrid has only $O(\sqrt B)$ cells on its boundary and each cell of the grid is accessed only a constant number of times, each of the $O(N/B)$ subgrids is loaded only $O(\sqrt B)$ times, causing $O(N/\sqrt B)$ \ios in total.

For the \io-efficiency under the confluence assumption, see the proof of Theorem~\ref{NaiveAccumulationZonRrealistic}.
\end{proof}

\subsection{Cache-aware flow accumulation using separators}\label{sec:cacheawareaccumulation}

We will now describe a slightly more involved algorithm that achieves an \io-efficiency of $O(\scan(N))$ even in the worst case. Unlike the previous algorithms, this algorithm is cache-aware: it needs to know $M$ and~$B$. The algorithm is based on separators much in the same way as the single-source shortest paths algorithm of Arge et al.~\cite{gridproblems} (but with a much smaller ``reduced'' graph). We present a cache-oblivious variant of our algorithm in Section~\ref{sec:cacheobliviousaccumulation}.

Let $z = \Theta(\sqrt M)$ be chosen such that a subgrid of size $z \times z$ fits in memory (while leaving a bit of space to store additional information about its boundary cells). Our cache-aware algorithm processes the grid in subgrids of size $z \times z$, where the boundary cells of each subgrid are shared with the neighbouring subgrids (except on the outer boundary of the grid). Let $S$ be the set of boundary cells of all subgrids. For any such subgrid $Q$, let $\id{interior}(Q)$ be the set of cells of $Q$ that do not lie on the boundary of $Q$, and let $R$ be the union of all subgrid interiors, that is, all cells except those in $S$. Given a subgrid $Q$ and a cell $c$ on its boundary, we define its local destination $\id{dest}(Q,c)$ as follows: if the out-neighbour of $c$ lies outside $Q$, then $\id{dest}(Q,c)$ is undefined, otherwise it is the first boundary cell of $Q$ that lies downstream of $c$ (if it exists). The algorithm now proceeds in three phases.

In the first phase, the rain that falls on each cell of $R$ is forwarded downstream to the first cell of $S$ that is reached (if any). The rain is only stored at that cell of $S$, not at the cells on the way. In the second phase, the rain that arrived at each cell $c$ of $S$, together with the rain that falls on $c$, is forwarded to all cells of $S$ downstream of $c$. In the third phase, the rain that falls on each cell of $R$ and the rain that arrived at each cell of $S$ is forwarded to all cells of $R$ that lie downstream of it, until another cell of $S$ is reached (or a cell without an out-neighbour). The result of these three phases is that the rain that falls on any grid cell $c$ is forwarded to all cells downstream of $c$ (see Figure~\ref{fig:cacheawarephases} for an illustration, Figure~\ref{fig:cacheawarecode} for pseudocode).

\begin{figure}
\centering
\includegraphics[width=\hsize]{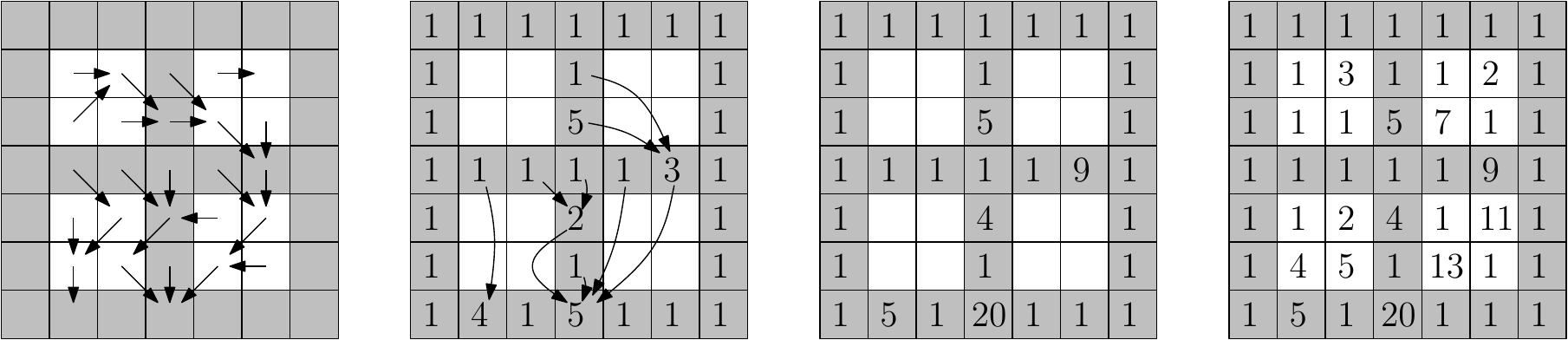}
\caption{Flow accumulation in three phases. The shaded cells are the cells of $S$. From left to right: input; contents of \id{SNeighbours} and \id{SFlowAcc} after the first phase; contents of \id{SFlowAcc} after the second phase; contents of \id{FlowAcc} after the third phase.}
\label{fig:cacheawarephases}
\end{figure}

\begin{figure}
\begin{codebox}
\Procname{$\proc{CacheAwareAccumulation}$}
\li Create file $\id{FlowAcc}$ with flow value (initially 1) and marking bit (initially not set) for each cell
\li Create file $\id{SNeighbours}$ with out-neighbour (initially undefined) for each cell of $S$
\li Create file $\id{SFlowAcc}$ with flow value (initially 1) and marking bit (not set) for each cell of $S$
\li \Comment phase one: push flow from subgrid interiors to $S$
\li \For each subgrid $Q$
\li \Do  run \proc{NaiveAccumulation} on $\id{interior}(Q)$ in memory,
\li      \hbox{}\quad\quad\quad\quad\quad\quad starting with zero flow on the cells of $Q \cap S$, and
\li      \hbox{}\quad\quad\quad\quad\quad\quad leaving $\id{out}(d)$ undefined for each cell of $Q \cap S$
\li      \For each cell $c$ on the boundary of $Q$
\li      \Do  store flow accumulation of $c$ in $\id{SFlowAcc}[c]$
\li           \hbox{}\quad\quad\quad\quad\quad\quad (adding it to any previously stored value)
\li           \If $\id{dest}(Q,c)$ is defined
\li           \Then $\id{SNeighbours}[c] \gets \id{dest}(Q,c)$
              \End
         \End
    \End
\li \Comment phase two: compute flow accumulation for $S$
\li Run line 2--8 of \proc{NaiveAccumulation} on $\id{SFlowAcc}$, using neighbour relations in $\id{SNeighbours}$
\li \Comment phase three: recompute flow in subgrid interiors, now including flow that comes in from $S$
\li \For each subgrid $Q$
\li \Do  \For each cell $c$ on the boundary of $Q$
\li      \Do $\id{FlowAcc}[c] \gets \id{SFlowAcc}[c]$
\li          \If $\id{out}(c)$ lies in $\id{interior}(Q)$
\li          \Then $\id{FlowAcc}[\id{out}(c)] \gets \id{FlowAcc}[\id{out}(c)] + \id{SFlowAcc}[c]$
             \End
         \End
\li      run line 2--8 of \proc{NaiveAccumulation} on $\id{interior}(Q)$, using the file $\id{FlowAcc}$,
\li      \hbox{}\quad\quad\quad\quad\quad\quad leaving $\id{out}(d)$ undefined whenever $\id{out}(d)$ lies on the boundary of $Q$
    \End
\end{codebox}
\caption{Pseudocode for cache-aware separator-based flow accumulation.}
\label{fig:cacheawarecode}
\end{figure}

\begin{theorem}\label{CacheAwareAccumulation}
With a tall cache of size $M \geq c B^2$, algorithm \textsc{CacheAwareAccumulation} needs only $O(\scan(N))$ I/O's.
\end{theorem}
\begin{proof}
In the first and the third phase, we process $O(N/M)$ blocks that each take $O(M/B + \sqrt M)$ \ios to read, can be processed completely in main memory, and---in the third phase---may take another $O(M/B + \sqrt M)$ \ios to write. The second phase runs a linear-time algorithm on a file of size $O(N/\sqrt M)$, of which each record is accessed a constant number of times in the first and the third phase. The total number of \ios is therefore $O(N/B + N/\sqrt M)$. By the tall-cache assumption, this is $O(\scan(N))$.
\end{proof}

\paragraph*{I/O-volume estimate for realistic values of $N$, $M$ and $B$}
\textsc{CacheAwareAccumulation} does not depend on the interplay between $B$, $M$ and $\gamma$ as much as the previous algorithms. Therefore it is possible to give a good estimate of the \io-volume. The bottleneck is the third phase. Assume we set up \id{SFlowAcc} and \id{SNeighbours} such that they first store, row by row, all rows that completely consist of cells of $S$, and then, column by column, the remaining cells of $S$. Then we need to choose $z$ such that the blocks containing $z$ rows of $z$ cells from \id{FlowAcc} and \id{FlowDir} fit in memory, plus 4 rows of $z$ cells from \id{SFlowAcc} and \id{SNeighbours}. In \id{FlowDir} we use 1 byte per cell, in the other files we use 8 bytes per cell. So $z$ must satisfy:\[
z \lceil 8z / B \rceil + z \lceil z/B \rceil + 4 \lceil 8z/B \rceil + 4 \lceil 8z/B \rceil \leq M/B
\]
A sufficient (but not necessary) condition is that $z$ is at most roughly $(\sqrt{1 + 8M/B^2} - 1) B/9$. This means that $z$ will typically be several thousands:
for $M = 2^{30}$ and $B = 2^{14}$ we would get $z = 8637$; for $M = 2^{30}$ and $B = 2^{16}$ we would get $z = 5330$.

Taking the number of subgrids times the number of rows times the number of blocks per row times the block size times two (for reading and writing), we find that the \io-volume for accessing \id{FlowAcc} in the third phase is now roughly
$N/z^2 \cdot z \cdot \lceil 8z/B\rceil \cdot B \cdot 2 \leq 16N + 2NB/z \approx N (16 + 18/(\sqrt{1 + 8M/B^2}-1))$.
From this we may subtract $8N$, since the blocks of \id{FlowAcc} do not need to be read from disk the first time they are accessed.
For accessing $\id{FlowDir}$ in the first and third phase (reading only) we get roughly
$N/z^2 \cdot z \cdot \lceil z/B\rceil \cdot B \cdot 2 \leq 2N + 2NB/z \approx N (2 + 18/(\sqrt{1 + 8M/B^2}-1))$.

The remaining \io is neglectible: for accessing $\id{SNeighbour}$ in the first phase we get roughly
$N/z^2 \cdot 4 \cdot \lceil 8z/B\rceil \cdot B \cdot 2 \leq 64N/z + 8NB/z^2$, which is roughly three orders of magnitude smaller than the amount of \io for accessing \id{FlowAcc}. The same is true for accessing \id{SFlowAcc} in the first and the third phase, and for accessing these files in the second phase: in practice $S$ will be small enough that the second phase can be done in main memory.

The \io-volume thus adds up to roughly $N (10 + 36/(\sqrt{8M/B^2+1}-1))$, dividing this by $9N$ (the input+output size) we get an ``overhead'' of factor $10/9 + 4/(\sqrt{8M/B^2+1}-1)$. For the example values of $M$ and $B$ just mentioned, the result ranges from 2.0 to 6.6.

\subsection{Cache-aware flow accumulation using separators with Z-order files}\label{sec:cacheawareaccumulationZorder}

The algorithm explained above depends on the tall-cache assumption. This is because the boundary of a square of $z \times z$ cells may cross $\Theta(z)$ blocks on disk, and without the tall-cache assumption, there is no guarantee that these will fit in memory for any $z = \Omega(\sqrt M)$. The need for the tall-cache assumption can be eliminated by using files in Z-order. Then any square $Q$ of $z \times z$ cells is contained in at most four squares of size $y \times y$, where $z/2 \leq y < z$, that are each stored contiguously on disk. The square $Q$ is thus stored in $O(4 \lceil y^2/B \rceil) = O(z^2/B)$ blocks, which will always fit for some $z = \Omega(\sqrt M)$. Thus we get:

\begin{theorem}\label{CacheAwareAccumulation}
Algorithm \textsc{CacheAwareAccumulation} on files in Z-order needs only $O(\scan(N))$ I/O's.
\end{theorem}

\paragraph*{I/O-volume estimate for realistic values of $N$, $M$ and $B$}
To get a good \io-volume in practice, we will now make sure that each subgrid $Q$ is stored consecutively in the Z-order file. To achieve this, we do not let neighbouring subgrids share rows and columns anymore, but make sure the subgrids are disjoint and their height is a power of two. Furthermore, we let the loops of line 5 and line 16 go through the subgrids in Z-order.

As a result, the first phase only scans each block of the \id{FlowDir} file once, and so does the third phase; in addition the third phase writes each block of \id{FlowAcc} once (reading is not necessary). There may be up to eight times more cells in $S$ (two times more because subgrids do not share boundary cells anymore, and four times more because we need to round the subgrid width down to a power of two). However, as explained above, the accesses to \id{SNeighbour} and \id{SFlowAcc} are so few that they will still be neglectible. Thus the total \io-volume becomes roughly N + N + 8N bytes, which is 1.1 times the input+output size.

\subsection{Cache-oblivious flow accumulation using separators}\label{sec:cacheobliviousaccumulation}

We can make a cache-oblivious version of the separator-based algorithm by using a hierarchy of subgrids and separators as follows. We consider $k+1$ levels, numbered 0 to $k$, where $k$ is the smallest integer such that the input grid is contained in a $(2^k+1) \times (2^k+1)$ grid. On level $i$, we consider subgrids of size $(2^i + 1) \times (2^i + 1)$, that each share their boundary cells with the neighbouring subgrid. Let $S_i$ be the set of boundary cells of all subgrids on level $i$, and let $R_i$ be the union of all subgrid interiors on level $i$. Note that $R_0 = \emptyset$, $S_0$ is the full grid, and $S_k$ is the boundary of the full grid. We can think of the subgrids as being organised in a tree~${\cal H}$, where the leaves correspond to subgrids of size $2 \times 2$, and each internal node on level $i$ corresponds to a subgrid $Q$ on level $i$ with four children, namely the subgrids on level $i-1$ that lie inside $Q$. The root of~${\cal H}$ corresponds to the full grid.

The algorithm now proceeds in two phases. In the first phase we do a post-order traversal of~${\cal H}$. When processing a subgrid $Q$ at level $i$, we forward the rain that falls on each cell $c$ of $S_{i-1} \cap \id{interior}(Q)$ to the first cell of $S_i$ downstream of $c$, that is, to the boundary of $Q$. In the second phase we traverse ${\cal H}$ in reverse order (which is therefore a pre-order traversal). When processing a subgrid $Q$ at level $i$, we forward the rain that arrived at each cell $c$ on the boundary of $Q$ to all cells of $\id{interior}(Q) \cap S_{i-1}$ that lie on the path that leads downstream from $c$ until it reaches another cell on the boundary of $Q$ (or a cell without an out-neighbour). These two phases together result in all rain that falls on any grid cell $c$ to be forwarded to all cells downstream of~$c$.

To implement this approach efficiently, we use the post-order traversal to create a stack $A$ of pointers from boundary cells to internal cells that can be retrieved during the reverse traversal. More precisely, when processing a subgrid $Q$ at level $i$ in the post-order traversal, we store on $A$ the coordinates of the first cell of $S_{i-1}$ that lies downstream of $c$, for each cell $c$ on the boundary of $Q$ that has an out-neighbour inside $Q$.

\begin{theorem}\label{CacheAwareAccumulation}
With a tall cache of size $M \geq c B^2$, separator-based cache-oblivious flow accumulation needs only $O(\scan(N))$ I/O's.
\end{theorem}
\begin{proof}
For simplicity we assume that the input grid is square.
We first analyse the number of pointers $P(K)$ that are stored on $A$ for a subtree of ${\cal H}$ rooted at a subgrid $Q$ of size $K$ (where a cell on the boundary of two/four subgrids counts for 1/2 or 1/4 to each of them). Since pointers are stored only for cells on the boundary of $Q$, we get $P(K) = O(\sqrt K) + 4 P(K/4)$, which solves to $P(K) = O(K)$.

Next we analyse the number of \ios $T(K)$ for the post-order traversal of a subtree of ${\cal H}$ rooted at a subgrid $Q$ of size $K$. For a sufficiently small constant $c$, we have $T(c \cdot M) = O(M/B + \sqrt M)$, since we would simply load $Q$ into memory, do all the necessary processing in memory, and then write $Q$ back to disk while  pushing $O(M)$ pointers for $Q$ and its subgrids on the stack. By the tall-cache assumption, we thus have $T(c \cdot M) = O(M/B)$. For subgrids $Q$ of size $K > c \cdot M$, we need to make the recursive calls, read the flow values for $O(\sqrt K)$ cells and the $O(\sqrt K)$ pointers stored on $A$ for the four subgrids of $Q$, and then use this information to forward flow to $O(\sqrt K)$ cells, compute $O(\sqrt K)$ pointers for $Q$ and push these onto $A$. Forwarding the flow and computing the pointers can easily be done in $O(\sqrt K)$ time and \ios, so that we get: $T(K) = O(\sqrt K) + 4 T(K/4)$.
With base case $T(c \cdot M) = O(M/B)$, this solves to $T(N) = O(N/\sqrt M + N/B) = O(\scan(N))$.
The analysis of the reverse traversal is similar to the analysis of the post-order traversal, so that we get $O(\scan(N))$ \ios in total.
\end{proof}

Although the cache-oblivious algorithms needs only $O(\scan(N))$ \ios, it is not as efficient as the cache-aware algorithm of the previous subsection. The cache-oblivious algorithm as just described is expensive in two ways. First, the number of pointers over all levels sums up to approximately $\frac83 N$. Except a small number on the highest levels, all of these pointers are written to and read from disk at least once. If a pointer is stored in 16 bytes (coordinates of source and destination), this amounts to an \io-volume of almost 10 times the input+output size. To alleviate this problem, an efficient implementation should use larger subgrids as the base case (for example $17 \times 17$ instead of $2 \times 2$). Second, we need to maintain flow accumulation values for boundary cells. Doing so directly in the output file is expensive, because on a vertical subgrid boundary, every cell will be in a different block. Although in the asymptotic analysis this works out, it causes a significant constant-factor overhead. Storing the grid in Z-order helps to some extent (as confirmed by our experiments, see Section~\ref{sec:evaluation}). Alternatively one could use the stack with pointers to store intermediate flow accumulation values (similar to $\id{SFlowAcc}$ in \proc{CacheAwareAccumulation}), but this would further increase the \io-volume for stack operations.

\subsection{Previous work: flow accumulation with time-forward processing}\label{sec:tfpaccumulation}

In \tfl and \tsm~\cite{terraflow,terrastream}, the following algorithm is used to compute flow accumulation. We first create a stream of cells in topological order. In principle it would suffice to sort the cells by decreasing elevation, but this would not work for flat areas. Therefore we will assume that after flooding, a flow routing phase has produced an additional file with a topological number for each cell. We now create a stream of cells that lists for each cell its location, its topological number and the topological number of its out-neighbour. We sort this stream by the first topological number.

We now apply a technique called time-forward processing: we maintain a priority queue in which flow values are stored with the topological number of the cell to which the flow is going. The algorithm processes all cells one by one in topological order; when processing a cell, it extracts its incoming flow from the priority queue, adds one unit of rain, writes the resulting total to an output stream together with the coordinates of the cell, and finally forwards the total flow to the out-neighbour by entering it in the priority queue with the out-neighbour's topological number as key. The algorithm finishes by sorting the flow accumulation values from the output stream into a row-by-row grid.

Using an \io-efficient priority queue, the algorithm clearly runs in $O(\sort(N))$ \ios~\cite{terraflow,gridproblems}

\paragraph*{I/O-volume estimate for realistic values of $N$, $M$ and $B$}
For an estimate of the \io-volume in a practical setting with a grid of size $N = 2^{32}$ and a memory size of $M = 2^{30}$, we consider two scenarios. In the optimistic scenario, we have block size $B = 2^{14}$, only 1/3 of the cells contain data, and the priority queue is completely maintained in memory at all times. In the pessimistic (but nevertheless quite realistic) scenario, we have block size $B = 2^{16}$, all cells contain data, and the records that pass through the priority queue are written to and read from disk once on average.

In the optimistic scenario, we start with scanning the input file with flow directions and the file with topological numbers, scanning the latter with a window of $3 \times 3$ cells to be able to retrieve topological numbers of neighbours. Assuming three rows of the grid fit in memory, this amounts to scanning 9 bytes per cell. While doing so, we generate a stream of cells with location, topological number, and topological number of the out-neighbour, which is fed to the first pass of a sorting algorithm. The sorting algorithm writes the partially sorted stream to disk (24 bytes per actual data cell = 8 bytes per grid cell). Since $N/3 \leq (M/B)^2$, the sorting algorithm needs only two passes; the second pass results in 16 bytes of \io per cell (24 for reading, 24 for writing, for 1/3 of the cells). The time-forward processing phase reads the sorted stream (24 bytes $\times$ 1/3 of the cells) and eventually outputs locations and flow accumulations for all data cells (16 bytes $\times$ 1/3 of the cells). The output then needs to be sorted: the first pass results in $10\frac23$ bytes of \io per cell (32 bytes per data cell). The second pass reads $5\frac13$ bytes per cell and writes 8 bytes per grid cell (flow accumulations only, but also for non-data cells). In total, we transfer $70\frac13$ bytes per grid cell. For a fair comparison, we do not count the input file with topological numbers towards the size of the input and output---after all, all other algorithms presented in this paper do not even need such a file. So $70\frac13$ bytes per grid cell amounts to 7.8 times the input+output size.

In the pessimistic scenario, we need three sorting passes, and the priority queue needs the disk. Filling in the details of the computation above, we get an \io-volume of 289 bytes per cell, which is 32 times the input+output size.

\subsection{A brief comparison of flow-accumulation algorithms}\label{sec:evaluation}
We implemented some of our algorithms and tested them on elevation models of the Netherlands and of the Neuse watershed in North Carolina, using a Dell Optiplex GX260 computer, equipped with a 3 GHz Pentium 4 processor and 1 GB of RAM, Ubuntu 7.04 (kernel 2.60.20-16) (installed on a 80 GB Samsung HD080HJ hard disk), and a 250 GB Seagate ST320506AS hard disk. We report the results for the largest data set: a grid of $70\,520 \times 50\,220$ cells modelling the Neuse basin; 35\% of the grid cells contain data.

Our results are shown in Table~\ref{tab:runningtimes}. For comparison, the authors of \tsm reported 455 minutes for flow accumulation of a grid of similar size, on hardware that appeared to be faster than ours~\cite{terrastream}.
Our cache-aware flow accumulation algorithm thus seems to be an order of magnitude faster. Note that the \io-volume estimates given in Section~\ref{sec:cacheawareaccumulation} and Section~\ref{sec:tfpaccumulation} would predict a difference of a factor four in the optimistic setting, where we assume that \tsm's time-forward processing would manage to do with sorting in two passes and keeping the priority queue in main memory. Although we have not run direct comparisons with the latest version of \tsm, our analysis indicates that the performance difference must remain significant because it is inherent to the characteristics of the different algorithms.

\begin{table}
\centering\def\arraystretch{1.1}
\begin{tabular}{|l|rr|rr|}
\hline
algorithm                      & \multicolumn{2}{c|}{using Seagate disk} & \multicolumn{2}{c|}{using both disks} \\
&  \multicolumn{1}{l}{time} & CPU usage                & \multicolumn{1}{l}{time} & CPU usage\\

\hline
``na\"ive'' row-by-row scan of row-by-row file & 111 min. & 22\% & & \\
``na\"ive'' Z-order scan of Z-order file    & 41 min. & 26\% & 34 min. & 32\% \\
cache-aware separator alg. on row-by-row file & 39 min. & 18\% & 25 min. & 26\% \\
\hline
\end{tabular}
\caption{Running times in minutes and CPU usage for flow accumulation of a grid of $3.5 \cdot 10^9$ cells.}
\label{tab:runningtimes}
\end{table}

Our results also indicate that our \io-na\"ive algorithms perform surprisingly well, especially when working on data in Z-order. At first sight our theoretical analysis seems to explain this: under the confluence assumption the \io-na\"ive algorithm on Z-order files needs only $O(\scan(N))$ \ios. However, if we try to estimate the \io-volume by filling in the constant factors in the computation of the asymptotic bound, then we end up with an \io-volume bound of dozens (using Z-order) or thousands (using row-by-row order) times the input+output size. It seems that the surprisingly good performance in practice must be due to the fact that for modest values of $c$ we can fit $c \sqrt B$ rows of the grid in memory. Thus the main scan can keep blocks in memory long enough so that each block only needs to be loaded once, and apparently the inner loop of the algorithm does not cause as many disruptions as the theoretical analysis might suggest. All things considered this means that these na\"ive algorithms are surprisingly fast, but we cannot rule out that their efficiency may depend on the ratio of $N$ and $M$ and on the characteristics of the terrain.

In contrast, the efficiency of our cache-aware algorithm is supported firmly by our theoretical analysis. Note that our theoretical analysis also indicates that the algorithm should be much faster still when working on Z-order files, requiring little more than two scans of the input file and a single scan of the output file (further experiments should confirm this). However, converting to Z-order takes time too (converting a grid of this size with 8 bytes per cell took 88 minutes). Whether it pays off to use files in Z-order may therefore depend on the context. Storing temporary files in the pipeline in Z-order seems to be a good idea; files that need to be processed by other software may better be kept in row-by-row order.

We also implemented the cache-oblivious algorithm, but until now this was significantly slower than the cache-aware algorithm (even when the subgrid size of the latter was not tuned optimally) so the cache-oblivious implementation would need to be optimised further to be competitive.

\section{Applying grid-based techniques to other parts of the pipeline}
\label{sec:otherstages}

\subsection{Flooding}
The flooding problem takes as input a file $\id{Elevation}$ that stores the elevation of each cell. A path in the grid is a sequence of cells such that each pair of consecutive cells on the path are neighbours of each other (each cell that is not on the boundary has eight neighbours). The flooding problem in its basic form is to compute to what elevation each cell should be raised, so that from each cell there is a non-ascending path to a cell on the boundary of the grid. If we define the height of a path as the elevation of the highest cell on the path, the problem is equivalent to determining, for each cell $c$, the height of the lowest path from $c$ to the boundary. The required output is a file $\id{Flooded}$ that stores these heights for each cell $c$.

\paragraph{Time-forward processing and I/O-na\"ive flooding}
In 2003 Arge et al.\ described two algorithms for the flooding problem~\cite{terraflow}. The first algorithm proceeds in three phases.
First we follow the path from each cell downhill until a sink, a cell without an out-neighbour, is reached. Each cell $c$ is labelled with a pointer $\pi(c)$ to the sink at the end of the path that goes downhill from $c$. This labelling constitutes a decomposition of the terrain into watersheds and is used to build a \emph{watershed graph} $W$. The watershed graph has one node for every sink, and for every pair of adjacent watersheds with sinks $u$ and $v$, the graph contains an edge $(u,v)$ with height equal to the lowest pass between $u$ and $v$. In other words, the elevation of $(u,v)$ is the minimum of $\max(\mathit{elevation}(c),\mathit{elevation}(d))$ over all pairs of neighbouring cells $c,d$ such that $\pi(c) = u$ and $\pi(d) = v$.
In the second phase an algorithm is run on $W$ to determine to what height each watershed should be flooded. In the third phase, we scan the complete terrain and replace each cell's elevation by the maximum of its original elevation and its watershed's flood height.

Arge et al.\ use time-forward processing for the first phase, so that the complete algorithm needs $O(\sort(N))$ \ios, assuming $W$ can be kept in main memory. Instead we could consider using an \io-na\"ive algorithm similar to the one described in Section~\ref{sec:naiveacc} for the first phase, and get $O(\scan(N))$ \ios under the confluence assumption. We found that in practice, this did not work well with row-by-row-ordered files. Using files in Z-order we could flood the Neuse grid mentioned before in 435 minutes: a running time of roughly the same order of magnitude as Terrastream's~\cite{terrastream}.

\paragraph{Separator-based algorithms}
The second algorithm sketched by Arge et al.\ seems quite similar to our cache-aware separator-based flow accumulation algorithm. As in Section~\ref{sec:cacheawareaccumulation}, the idea is to first process subgrids of $\Theta(\sqrt M) \times \Theta(\sqrt M)$ cells that fit in memory. For each subgrid $Q$ we compute a `substitute' graph that encodes the lowest-path heights between each pair of cells on the boundary of $Q$. In the second phase of the algorithm we would combine the substitute graphs for all subgrids into one graph, which is used to compute the lowest-path heights from each cell of $S$ to the boundary of the grid, where $S$ is the set of boundary cells of all subgrids. Finally, in the third phase each subgrid $Q$ is processed again, now to compute the lowest-path height from each cell of $Q$ to the boundary of the grid, using the previously computed lowest-path heights for the cells on the boundary of $Q$.

The key to success is the size of the substitute graphs. The algorithm by Arge et al.~\cite{gridproblems} for single-source \emph{shortest} paths works in a similar way, but needs substitute graphs of size $\Theta(z^2)$ to encode the shortest-path distances between all pairs of cells on the boundary of a $z \times z$ subgrid. However, for \emph{lowest} paths substitute graphs of size $\Theta(z)$ suffice. Such a graph can be created from the elevations and neighbour relations in a subgrid $Q$ as follows. Consider the subgrid $Q$ as a graph, whose edges represent the neighbour relations in the grid, where each edge has elevation equal to its highest vertex. Compute the lowest paths to the boundary of $Q$ for all cells in $\mathit{interior}(Q)$. Next, contract all directed edges $(u,v)$ of those lowest paths one by one, replacing each undirected edge $(w,u)$ of the graph with an edge $(w,v)$ with elevation $\max(\mathit{elevation}(w,u),\mathit{elevation}(u,v))$. Whenever there are multiple edges between the same pair of vertices, keep only the edge with the lowest elevation. This results in a substitute graph whose vertices are the boundary vertices of $Q$ and which preserves the lowest-path heights between each pair of vertices. Since the substitute graph thus constructed is planar, it has size $\Theta(z) = \Theta(\sqrt M)$.

We implemented a separator-based flooding algorithm as described above, and found that it could process the Neuse grid in 146 minutes, using row-by-row-ordered files on one disk, or in 132 minutes on two disks. As with our flow accumulation algorithm, we believe further efficiency gains could be achieved by using files in Z-order.

\paragraph{Partial flooding?}
It should be noted that a direct comparison between our algorithms and \textsc{Terra-Stream} cannot be made. \tsm offers the functionality of \emph{partial flooding}: eliminating only insignificant depressions while not flooding major depressions. Our algorithms do not do this. A major open question is therefore how the grid structure could be exploited to design an algorithm that can do very fast \emph{partial} flooding.

\subsection{Other parts of the pipeline}
\tsm's hierarchical watershed labelling algorithm~\cite{pfafstetter,terrastream} uses time-forward processing to pass labels upwards into the river network. This is not very different from how time-forward processing is used for flooding or for flow accumulation, and one may expect that grid-based algorithms (\io-na\"ive or cache-aware) may help here too. Open questions include whether grid-based algorithms, maybe together with assumptions on realistic terrains, could help to simplify and speed up flow routing on flat areas and to do flow accumulation with multiple-flow directions, an approach where each cell sends water to all of its lower neighbours instead of just one of them.

\section{Conclusions and remaining work}
We have shown that certain hydrological computations on terrain data may be sped up by an order of magnitude by exploiting the grid structure of the data and/or by storing grids in Z-order rather than row-by-row order. A striking result is that one of the algorithms with the best performance is actually so simple that it is almost na\"ive. Some of the most prominent questions that remain unanswered at this point are the following.

\emph{What would be typical values for the confluence constant?}
It would be interesting to design an algorithm that can compute the confluence constant for any given grid terrain. Then we may investigate to what extent the confluence constant is indeed independent of the sampling density of a terrain, and what are typical values for the confluence constant for different types of terrains.

\emph{How would the cache-aware separator-based algorithm perform on files in Z-order?}
Unfortunately we did not have time to implement and run these tests yet, and we hope to be able to do so some time in the future.

\emph{Can we exploit the grid structure to design equally efficient flow routing algorithms?}
If yes, then the complete part of the pipeline from elevation model to hierarchical watershed labels could probably be sped up tremendously. In that case, even if files in row-by-row order would be desired at the input and output end of the pipeline, it could pay off to convert them to Z-order and to use files in Z-order for the intermediate stages. A disadvantage of our current algorithm is that multiple-flow direction models cannot be handled, but for hierarchical watershed labelling such models cannot be used and single-flow direction models---those handled by our algorithms---are exactly what is needed.

\paragraph*{Acknowledgements}
The authors thank Andrew Danner for providing test data and assistance with \tsm.
We thank Laura Toma and the students of the 2006 class of \io-efficient algorithms
at the TU Eindhoven for inspiring discussions.

\small
\bibliographystyle{abbrv}

\end{document}